\documentclass[10pt, a4paper, oneside, reqno]{amsart}

\makeatletter
\def\@settitle{\begin{center}%
		\baselineskip14\p@\relax
		\normalfont\LARGE\scshape\bfseries
		\@title
	\end{center}%
}
\makeatother

\makeatletter

\def\subsection{\@startsection{subsection}{2}%
	\z@{.5\linespacing\@plus.7\linespacing}{.5\linespacing}%
	{\normalfont\bfseries}}

\def\subsubsection{\@startsection{subsubsection}{3}%
	\z@{.5\linespacing\@plus.7\linespacing}{.5\linespacing}%
	{\normalfont\itshape}}

\usepackage[usenames, dvipsnames]{color}
\definecolor{darkblue}{rgb}{0.0, 0.0, 0.45}

\usepackage[colorlinks	= true,
			raiselinks	= true,
			linkcolor	= darkblue, 
			citecolor	= Mahogany,
			urlcolor	= ForestGreen,
			pdfauthor	= {},
			pdftitle	= {},
			pdfkeywords	= {},
			pdfsubject	= {},
			plainpages	= false]{hyperref}

\usepackage{dsfont,amssymb,amsmath,graphicx} 
\usepackage{amsfonts,dsfont,mathtools,mathrsfs,amsthm} 
\mathtoolsset{showonlyrefs=true}
\usepackage[amssymb, thickqspace]{SIunits}
\usepackage{fancyhdr,setspace}

\usepackage[numbers,compress]{natbib}
\usepackage{nicefrac}       
\usepackage{microtype}      
\usepackage{tabto}

\usepackage{enumerate,enumitem}
 \usepackage{caption}

\usepackage{bbm}

\allowdisplaybreaks
\date{\today}

\newcommand{\first}{t_{1,n,k}}
\newcommand{\Second}{t_{2,n,k}}
\newcommand{\seconds}{t_{2,n,k}^s}
\newcommand{\PI}{\pi_{n,k}^s}
\newcommand{\Xn}{x_{n}(t)}
\newcommand{\Xni}{x_{n}^{i}(t)}
\newcommand{\Yn}{y_{n}(t)}
\newcommand{\X}{x^{i}_{n}(t)}

\newcommand{\XT}{x^{i}_{n}(t)}
\newcommand{\YT}{y^{i}_{n}(t)}

\newcommand{\SUMN}{\sum\limits_{n=1}^{N}}

\newcommand{\TK}{\mathcal{T}_{k,n}}

\newcommand{\E}{\mathbb{E}}
\newcommand{\NN}{n\in\mathcal{N}}

\graphicspath{{images/}}
\DeclareMathOperator{\sign}{sign}

\usepackage{url}            

\usepackage{algorithm}
\usepackage{algorithmic}

 \usepackage[foot]{amsaddr}

\usepackage{subfig}

\usepackage{algorithmic}
{

{
		\theoremstyle{plain}
		\newtheorem{theorem}{Theorem}
	}
	{
		\theoremstyle{plain}
		\newtheorem{assumption}{Assumption}
	}
	{
	\theoremstyle{plain}
	\newtheorem{lemma}{Lemma}
}
	{
		\theoremstyle{plain}
		\newtheorem{definition}{Definition}
	}
	{
		\theoremstyle{plain}
		
	}
{
		\theoremstyle{plain}
		\newtheorem{remark}{Remark}
	}
{
		\theoremstyle{plain}
		
	}

\title[]{Multi-agent maintenance scheduling based on the coordination between central operator and decentralized producers in an electricity market}

\author[1]{Pegah Rokhforoz$^1$}
\author[2]{Blazhe Gjorgiev$^2$}
\author[2]{Giovanni Sansavini$^2$}
\author[3]{Olga Fink$^{3}$}
\address[1]{Chair of Intelligent Maintenance Systems, ETH Zurich, Switzerland, and School of Electrical and Computer Engineering, University of Tehran, Tehran, Iran.}
\address[2]{Reliability and Risk Engineering Laboratory, Institute of Energy Technology, Department of Mechanical and Process Engineering, ETH Zurich, Switzerland.}
\address[3]{Chair of Intelligent Maintenance Systems, ETH Zurich, Switzerland. Corresponding author, email address:ofink@ethz.ch.}

\begin{document}

\maketitle

\begin{abstract}
Condition-based and predictive maintenance enable early detection of critical system conditions and thereby enable decision makers to forestall faults and mitigate them. However, decision makers also need to take the operational and production needs into consideration for optimal decision-making when scheduling maintenance activities. Particularly in network systems, such as power grids, decisions on the maintenance of single assets can affect the entire network and are, therefore, more complex. 

This paper proposes a bi-level multi-agent decision support system for the generation maintenance decision (GMS) of power grids in an electricity market in the context of predictive maintenance. The GMS plays an important role in increasing the reliability at the network level. The aim of the GMS is to minimize the generation cost while maximizing the system reliability. The proposed framework integrates a central coordination system, i.e. the transmission system operator (TSO), and distributed agents representing power generation units that act to maximize their profit and decide about the optimal maintenance time slots while ensuring the energy balance. We derive the optimal strategy of the agents that are subject to predictive maintenance via a distributed algorithm, through which agents make optimal maintenance decisions and communicate them to the central coordinator, i.e. the TSO. The TSO decides whether to accept the agents' decisions by considering market reliability aspects and power supply constraints. To solve this coordination problem, we propose a negotiation algorithm using an incentive signal to coordinate the agents' and the central system's decisions, such that all the agents' decisions can be accepted by the central system. We demonstrate the effectiveness of the proposed algorithm with reference to the IEEE 39 bus system.
\end{abstract}

\noindent \textit{Indices}

\begin{description}
\item [$t$] \quad Index of hour.

\item [$n$] \quad  Index of generation units (agents).

\item [$s$] \quad Index of scenario path.

\item [$j$] \quad Index of nodes.

\end{description}

\noindent \textit{Constants}

\begin{description}

\item [$T_{n}$]  Fixed planning time for maintenance of agent $n$.

\item [$N$]  Number of generation units (agents).

\item [$S_{n}$] Number of path scenarios for the second threshold of agent $n$.

\item [$J$] Number of nodes.

\item [$\Gamma_{j}$] Set of nodes directly connected to node $j$.

\item [$P(t)$] Electrical market price at time \textit{t}$[\frac{\$}{MW}]$.

\item [$q_{n}^{\max}$] \quad Maximum power output of unit $n$ $[MW]$.

\item [$q_{n}^{\min}$] \quad Minimum power output of unit $n$ $[MW]$.

\item [$C_{n}(t)$] Generation cost of unit $n$ at time $t$ $[\$]$.

\item [$L_{j}(t)$] \quad   Load of system at node $j$ at time $t$ $[MW]$.

\item [$B_{j,r}$] Susceptance of line $jr$ [per unit].

\item [$F_{j,r}$] Transmission capacity of line $jr$ $[MW]$.

\item [$\first$]  First threshold of unit $n$.

\item [$\Second$]  Second threshold of unit $n$.

\item [$r_{n,k}$] Repair time of unit $n$.

\item [$\seconds$] Second threshold scenario path $s$ of unit $n$.

\item [$\PI$] Probability of second threshold scenario path $s$ of unit $n$.

\item [$\alpha_{n}$] Positive constant of unit \textit{n} [$\$/h$].

\item [$\gamma_{n}$] Positive constant of unit \textit{n} [$\$$].
\end{description}

\noindent \textit{Decision variables} 

\begin{description}

\item [$x_{n}(t)$] \quad Maintenance decision of unit $n$ at time $t$.

\item [$y_{n}(t)$] \quad Maintenance decision of central system for unit $n$ at time $t$.

\item [$q_{n}(t)$] Power output of unit $n$ at time $t$ $[MW]$.

\item [$\theta_{j}(t)$] Angle of node $j$ at time $t$ $[rad]$.

\item [$u_{n}(t)$] Commitment decision of Agent $n$ at time $t$.

\end{description}
{\textit{Notations.}}  The function $\sign(\cdot)$ represents the standard sign function as 

$\sign(x)=
    \begin{cases}
      1, & x>0 \\
      0,&   x=0\\
      -1,& x<0.
    \end{cases}$ 
    
    Let us consider $\zeta$ as a stochastic variable, we denote by $\E_{\zeta}{\{f(\cdot,\zeta)}\}$ the expected value of function $f(\cdot,\zeta)$ with respect to $\zeta$. We define the column augmentation of $x_{n}(t)$ for $t={\{1,\cdots,T_{n}}\}$  as $\mathbf{x_{n}}:=\textbf{col}(x_{n}(1),\cdots,x_{n}(T_{n})):=[x_{n}(1),\cdots,x_{n}(T_{n})]$.

\section{Introduction}

Traditionally, maintenance actions in power systems have been scheduled in a preventive way, typically with long planning decision horizons. Maintenance intervals have to strike a balance between system availability, the impact of failures and maintenance costs \cite{Kralj_1988}. To this aim, condition-based and predictive maintenance enable to perform maintenance activities only when required, based on the condition of the system and the predicted evolution of its degradation \cite{Bangalore_2016,Dieulle_2003,Yang_2008}. However, this decreases the planning decision horizons significantly due to the short reaction times of these algorithms and entails a complex maintenance scheduling because of the strong coordination required by the different decision makers. Indeed, the decisions of individual stakeholders affect the operations of the whole power network. 

The goal of maintenance scheduling is to maximize the reliability of the system. One major challenge associated with maintenance scheduling in network systems, like power grids, is that the decision makers need to comply with system constraints. The main constraint of the power grid in the electrical market is the instantaneous balance of power. Making a trade-off between maximizing the reliability and balancing the power demand plays a significant role in the  maintenance scheduling of generating units. Several research studies have addressed the trade-off between maximizing reliability and load balancing using an optimization approach, agent-based modelling and incentive signals solutions \cite{Shahidehpour_2012,Moslehi_2010}.

In this paper, we propose a bi-level negotiation algorithm in the context of predictive maintenance between the central coordinating system and decentralized power generating agents using incentive signals. Our proposed framework is specifically targeting predictive maintenance decisions, which enable agents to perform maintenance before the failure occurs based on the predicted remaining useful life and its uncertainty. Hence, the reliability of the system is preserved and corrective maintenance costs are avoided. However, the agents have only a limited time window to make their decisions. Within the framework of predictive maintenance, we assume that the remaining useful lifetime (RUL) of a generating unit can be predicted and the associated uncertainty estimated. We develop a model for the agents' objective functions, which is based on the expected deterioration cost and on the revenue that the agents obtain from the energy production. In addition, to ensure power balance, we propose a negotiation algorithm where agents submit their optimal maintenance time slot to the central system. Agents' decisions are rejected by the central system if the load fulfillment cannot be satisfied during the time period of the maintenance. We develop a maintenance optimization model for the central system which strikes a balance between the proximity to the end of life prediction of the individual plants and the need to satisfy the load balance at the system level. If the agents' decisions cannot be accepted, the central system sends an incentive signal to encourage the agents to adjust their proposition of the maintenance time slot before failure. The procedure iterates until the convergence occurs, i.e. the adjusted decisions of all agents can be accepted by the central system. The proposed mechanism provides a weak budget balance during the negotiation process, i.e. during negotiations, since the agents pay a penalty to the central system, its revenue is positive. This indicates that the mechanism of incentive signals by the central system does not add overheads to the market mechanism. Indeed, this is the main feature that a mechanism design solution (negotiation algorithm) should possess \cite{Kosenok_2008}. Moreover, in the convergence point of the algorithm, the proposed algorithm is fully budget balanced which means that the revenue of the central system is zero. Furthermore, our simulation results show that the expected rewards of the agents are positive during the convergence process. Hence, the proposed induced mechanism is individually rational for each of the agents. This means that the agent will participate voluntarily in the mechanism under the assumption of rationality.


In summary, the proposed framework comprises three main contributions, which go beyond the state of the art, namely: 

\begin{enumerate}
\item  We propose a decentralized framework using a bi-level negotiation algorithm that is suitable for predictive maintenance scheduling, where agents perform maintenance before they fail.

\item The proposed induced mechanism is weakly budget balanced during the negotiation and is fully budget balanced in the convergence point of the algorithm. In other words, at the convergence point of the algorithm, the central system does not need to pay the money to the agents to encourage them to change their decisions.

\item The proposed approach is individually rational and agents participate in the mechanism voluntary.
\end{enumerate}

To the best of our knowledge, this is the first time that a multi-agent  maintenance scheduling framework based on the coordination between central operator and decentralized producers in an electricity market is proposed in the context of predictive maintenance using the negotiation algorithm approach which is budget balanced and individually rational.

The remainder of the paper is organized as follows. A review of the related work is given in Section \ref{sec:Related}. The basics on remaining useful life (RUL) prediction are introduced in Section \ref{sec:RUL}. We present the framework of the negotiation algorithm in Section \ref{sec:framework}. The problem formulation which consists of the agents' and central system's objective function is introduced in \ref{sec:problem}. Section \ref{sec:coordinated} represents the coordination procedure among agents and central system using an incentive signal. The case study is introduced in Section \ref{sec:case}. The simulation results are shown in Section \ref{sec:simulation}. The concluding remarks are made in Section \ref{sec:conclusion}.

\section{Related work}
\label{sec:Related}
\textbf{Centralized generation maintenance scheduling.} The optimization approach for the maintenance schedule of power generating units in an electrical market system has been implemented in different ways. A method for optimal maintenance scheduling based on genetic algorithms has been proposed in \cite{Volkanovski_2008,Samuel_2015}. A meta-heuristic approach for the optimal maintenance scheduling of generation units using reliability-based objective function is presented in \cite{Dahal_2007}. Furthermore, a method for joint optimization of generation scheduling and preventive maintenance has been elaborated in \cite{Xiao_2016,Sadeghian_2019}. The goal of the optimization model proposed in \cite{Xiao_2016} is to minimize the total cost including production cost, preventive maintenance cost, minimal repair cost for unexpected failures and tardiness cost. The authors of \cite{Sadeghian_2019} propose a multi-objective optimization model to maximize the profits of selling electricity and maintaining the system reserves. The generation scheduling from the point of view of the central system is formulated in \cite{Ghazvini_2012}, which can be solved using linear mixed integer programming approach. All of these research studies achieve an optimal solution of the maintenance optimization in a centralized manner. In other words, the generating units do not decide on their maintenance independently. The authors of \cite{Abiri_2012} address the predictive maintenance scheduling problem using a centralized optimization approaches which works based on the decision tree and mixed integer linear programming. Their maintenance schedules are solely determined by the central system. These centralized approaches may result in high computational costs and the central system needs to know the private information of all agents.

\textbf{Decentralized generation maintenance scheduling.}
In addition to solving the maintenance scheduling as a central optimization problem, several research studies formulate the maintenance scheduling problem in a decentralized way, where agents decide on their maintenance scheduling based on individual objective functions. In \cite {Ren_2019}, an agent-based approach is proposed to maximize the system reliability while considering the load balance as a constraint. A bi-level optimization approach is developed in \cite{Li_2005} where agents maximize their rewards in the first level and the central system clears the market in the second level. In \cite{Kuznetsova_2014},  the authors propose a decentralized robust optimization approach which leads to an increase of system performance and reliability. A stochastic maintenance scheduling is proposed in \cite{Ghazvini_2013}, where agents choose their maintenance decisions and the TSO makes a final decision by considering the system's reliability and its constraints. All of these research studies do not consider any negotiation or coordination mechanism between the central system and the agents.

\textbf{Generation maintenance scheduling based on incentives or penalty signals.}
Several research studies have also addressed the generation maintenance scheduling based on incentives or penalty signals. The authors of \cite{Conejo_2005} propose an incentive signal which ensures an appropriate level of reliability. In this paper, the central system obtains the maintenance scheduling using an optimization approach to maximize the reliability and minimize a cost function. Moreover, each agent seeks to maximize its objective function and minimizes its maintenance cost that conflicts with the goal of the central system. This challenge is solved using a coordination mechanism based on an incentive signal. The maintenance scheduling of generating units based on game theory is proposed in \cite{Min_2013}. The authors consider a competition among agents using game theory and develop a coordination mechanism using an incentive signal to align the objective function of the agents with that of the central system. In \cite{Feng_2009}, a novel mechanism that balances between the benefits of the agents and the system reliability is proposed. In this paper, the agents submit a set of maintenance bidding costs to the central system by considering unexpected events of unit failures. The central system schedules the maintenance which satisfies the system's energy demand. Then, the central system sets the final expenditure to make a coordination mechanism. In \cite{Wang_2016}, a further coordination mechanism to maintain the central system reliability while maximizing the benefits of agents is proposed. One of the limitations of the previously proposed mechanisms is that not all of the mechanisms are budget balanced. Furthermore, the proposed mechanism in these works are not in the context of predictive maintenance and they just consider the maintenance cost in their objective functions which is the main difference with the proposed mechanism where the agents seek to minimize its predictive maintenance cost.

\textbf{Generation maintenance scheduling in the context of predictive maintenance.}
 The focus of the previous research studies has been mostly on the combination of corrective and preventive maintenance. Predictive maintenance has only been considered in the maintenance scheduling problems in a limited way. There are few research studies that address the generation maintenance scheduling considering predictive maintenance modelling. The authors of \cite{Yildirim_2016, Yildirim1_2016} propose a mixed-integer optimization model for generation maintenance scheduling using the information of generators' health data. 
These researchers implement the centralized approach which is the main difference with our paper.

\section{Predicting the remaining useful life (RUL)}
\label{sec:RUL}
The prediction of the remaining useful life is at the core of predictive maintenance applications. Remaining useful life (RUL) is the amount of time, in terms of operating hours, cycles, or other measures in which the component will continue to meet its design specification \cite{lee2014prognostics}.  RUL estimation is an essential part of prognostics and health management (PHM) of industrial systems. Prognostics estimate the lifetime of the  specific component in its specific operating  environment and are not solely based on the average system behaviour.
 Several directions have been proposed for estimating the RUL that fall into the three main categories \cite{lee2014prognostics, fink2020data}:
(a) model-based or physics-based approaches, i.e. combining typically two parts: modelling the underlying physics of a component/subsystem and modelling physics of damage propagation mechanisms; (b) data-driven approaches that are based on condition monitoring data \cite{Si_2012, kraus2019forecasting}; (c) knowledge-based approaches that rely on domain expert judgements.

Two cases can be distinguished in RUL prediction: 1) the case that RUL prediction is only determined by a continuous gradual deterioration of the system; 2) the case that the condition monitoring system detects the onset of an incipient fault that causes an accelerated deterioration of the system that dominates the RUL. For the more challenging second case, the RUL prediction typically integrates two steps, the detection of the fault onset and a subsequent prediction of the RUL based on the predicted operating profile. 

Since the estimation of the RUL involves predicting the future behavior of industrial assets, it implies several sources of uncertainty that influence the future prediction. Therefore, it is
rarely feasible to estimate the RUL with complete certainty and RUL cannot be considered as deterministic \cite{sankararaman2013remaining}. 

The prediction of RUL, therefore, typically also involves the quantification of the uncertainty associated with the specific prediction \cite{sankararaman2013remaining, fink2020potential, sankararaman2013analytical}. Different approaches have been proposed to estimate the uncertainty of the RUL prediction and combine the different sources of uncertainty \cite{saxena2010evaluating}. Indeed, the estimation of the RUL does not prevent the failure of the component before the predicted end of life (EoL). In this paper, we assume that we can detect the onset of a fault and can subsequently predict the RUL with the associated uncertainty. We can, therefore, assume that the probability distribution of the estimated uncertainty is known.

\section{The proposed negotiation mechanism}
\label{sec:framework}
The proposed bi-level multi-agent system for the maintenance decision support aims at minimizing the generation cost while maximizing the system reliability in terms of generation adequacy. The framework integrates a central coordination system, the transmission system operator (TSO) and distributed agents representing power generation units that act to maximize their profit and decide about the optimal maintenance time slots while ensuring the power demand balance. 

The mechanism  consists of four main steps which enforce the reliability and the power demand balance: 
\begin{enumerate}
\item Each agent obtains its predictive maintenance modelling ($DC_{n}$, $\NN$, which is defined in the next section) by considering the fault onset time $t_{1}$ and failure time $t_{2}$. $t_{1}$ is obtained using the condition monitoring data and is certain, but $t_{2}$ is predicted based on the prognostic algorithm and is a random variable with a known distribution function.

  \item Each agent considers its failure time as a stochastic random variable and decides on the maintenance time slot for a fixed period of time ($x_{n}$, $\NN$), depending on the expected rewards that it obtains by generating the power and the expected cost of maintenance. Based on this, the agents submit their decisions (bids) to the central system.
 
  \item The central coordinating system is responsible for the reliability at the network level and the load balance. Hence, if accepting all of agent decisions does not guarantee the power balance, the central system selectively accepts the bids of the agents that are close to the failure time and rejects other bids which is expressed as ($y_{n}, \NN$) in Figure \ref{fig:negotiation}. 
  
  \item Since the agents will fail if no maintenance is performed within the specified time period and incur into failure costs, the central coordinating system encourages them to change their decisions by providing them with an incentive signal ($I_{n}, \NN$) which is explained in Section \ref{sec:coordinated}.

The information flow among agents and central system is shown in Figure  \ref{fig:negotiation}.

\begin{figure}[htb!]
\centering
{\onecolumn\includegraphics[scale=0.9]{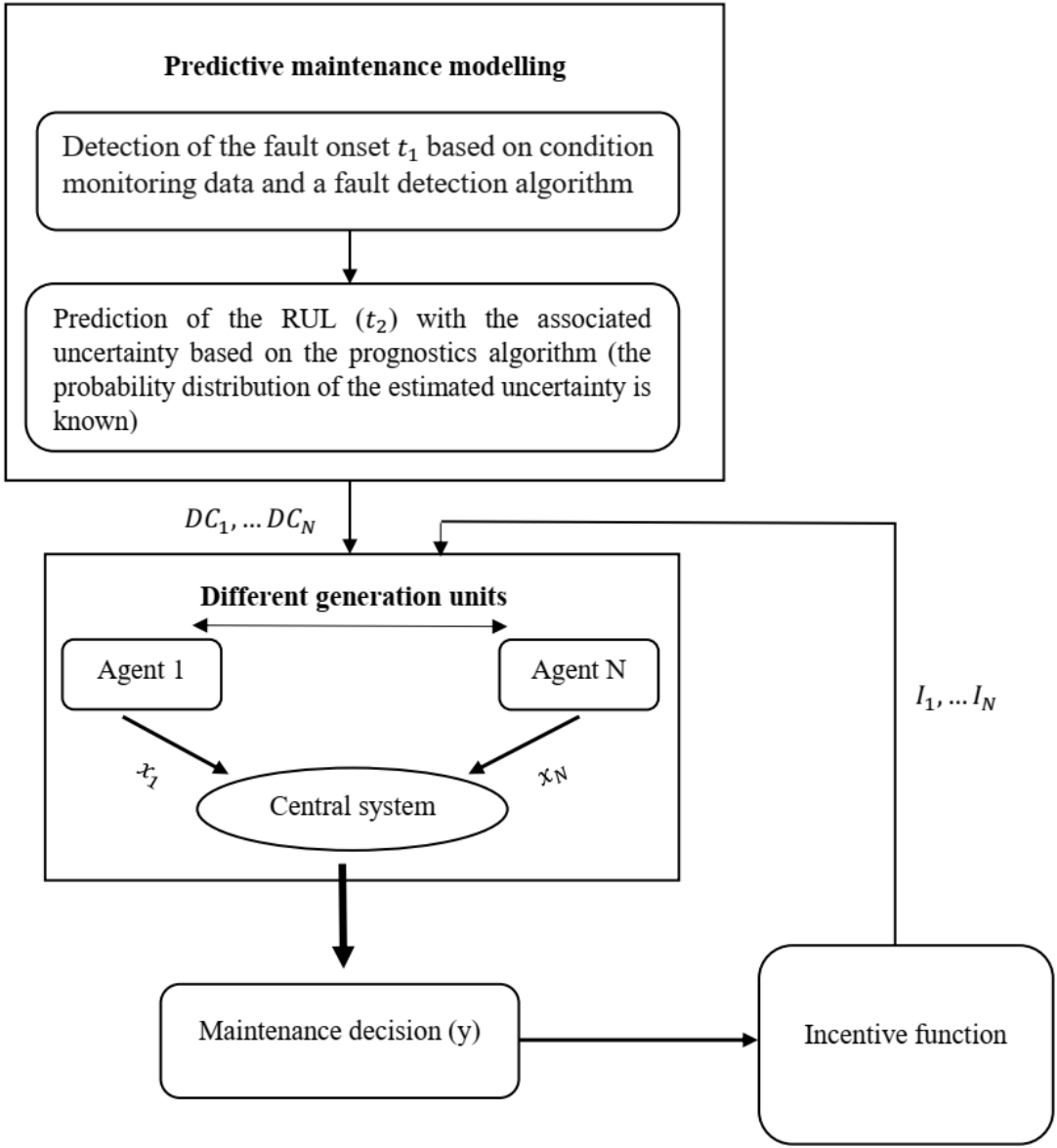}}
\caption{Negotiation schematic between power generating agents and the central coordinating system.}
  \label{fig:negotiation}
\end{figure}  

\end{enumerate}

\section{Problem formulation}
\label{sec:problem}

\subsection{Stochastic deterioration cost}
In this research study, we assume a setup where generation units are equipped with condition monitoring devices, which can detect the fault initiation. We assume that the failure behaviour of the components is not dominated by continuous gradual degradation but rather by faults that can be detected by the condition monitoring system. After the detection, the remaining useful lifetime is predicted based on the current system state, the predicted evolution of the detected faulty system condition and the future operating profile. Hence, the fault detection is considered deterministic since it happens in real time and
the prediction of the RUL is only initiated once the fault has
been detected. The maintenance action should be then performed within the detected fault initiation and the predicted failure time. To model the increasing deterioration and the increasing failure probability of the component after the detection of the fault, we introduce a fault progression penalty.

 In order to model the deterioration cost, we consider $\mathcal{T}_{n}={\{1,\cdots,T_{n}}\}$ as the maintenance scheduling time horizon of the agent $n$, and we define $t_{1,n}$ the fault initiation of agent $n$, and $t_{2,n}$, as the end of life of agent $n$ which is a random variable with a known distribution function and is equivalent to the failure time or the end of life of the generation unit $n$, $\NN$. Since the fault progression is typically nonlinear, we model the penalty costs between the fault detection and the end of life as an exponential function. $t_{1,n}$ is the starting time to impose a fault progression penalty on agent $n$, $\NN$. The fault progression cost can be modelled as follows:

\begin{equation}
\begin{aligned}
    DC_{n}=\sum\limits_{t\in\mathcal{T}_{n}}\E_{t_{2,n}}{\Big\{\big(1-\Xn\big)\alpha_{n}e^{(t-\first)}\big(\sign{(t-\first)}-{{\sign(t-\Second)}\big)\Big)}\Big\}}.\\
    \label{eq:agent's deterioration cost}
\end{aligned}
    \end{equation}
    
    Equation \ref{eq:agent's deterioration cost} implies that if agent $n$ does not perform a maintenance action after the fault initiation $t_{1,n}$, then it receives a penalty cost. Moreover, the unit will fail if the agent does not perform maintenance before the failure time. In this case, no deterioration cost is imposed on the agent after that point. In other words, the agent is penalized by the fault progression cost between $t_{1,n}$ and $t_{2,n}$, $\NN$, which is modeled by $\big(\sign{(t-\first)}-{\sign(t-\Second)}\big)$. Furthermore, since $t_{2,n}$, the failure time, is subject to uncertainty and is considered as a random variable, we use the expected fault progression cost for agent $n$, $\NN$. Since the agents seek to minimize their fault progression cost, the penalty function \eqref{eq:agent's deterioration cost} based on the value of $\alpha_{n}$ encourages agent $n$, $\NN$, to perform maintenance before the failure time. It is determined by the central system and is the fixed failure cost term independent of time.

\subsection{Agents' objective function}
The generating units determine the maintenance time slots by maximizing their expected reward functions. The agents are price-takers, i.e. they generate revenues by providing the capacity to meet the power demand. Each agent's objective function comprises two parts: 1) the sum of the expected gain from the generated power over the maintenance scheduling time horizon, which is computed by subtracting the sum of the generation costs from the revenues, and 2) the expected fault progression cost.

We model the expected reward function of agent $n$, $\NN$, as follows:

\begin{equation}
\begin{aligned}
    R_{n}=\sum\limits_{t\in\mathcal{T}_{n}}\E_{t_{2,n}}{\Big\{\big(1-\Xn\big)\Big(\frac{1}{2}\big(1-\sign(t-\Second)\big)\big(P({t})q_{n}({t})-C_{n}({t})\big)\Big)}\Big\}.\\
    \label{eq:agent's reward}
\end{aligned}
    \end{equation}

Equation \ref{eq:agent's reward} implies that if agent $n$, $\NN$, performs maintenance at time $t$ ($x_{n}(t)=1$) its revenue is zero, because  it cannot produce power. Furthermore, when agent $n$, $\NN$, does not perform maintenance before its failure time, it will fail in the sense that $\big(1-\sign(t-\Second)\big)=0$ and its generation and revenue become zero as well. It is important to note that the failure time is a random variable. Therefore, the sum of the expected rewards from the agent's generated power is considered.

For agent $n$, $n\in\mathcal{N}$, the decision-making process can be formulated as an optimization problem:
\begin{equation}
\begin{aligned}
    \max\limits_{\mathbf{x_{n}}}&\sum\limits_{t\in\mathcal{T}_{n}}\E_{t_{2,n}}{\Big\{\big(1-\Xn\big)}\Big(\frac{1}{2}\big(1-\sign(t-\Second)\big)\big(P({t})q_{n}({t})-C_{n}({t})\big)\\
    &-\alpha_{n}e^{(t-\first)}\big(\sign{(t-\first)}-{{\sign(t-\Second)}\big)\Big)}\Big\}\\
   \text{s.t.}&\quad\text{C}_{1}:\quad\Xn\in{\{0,1}\},\quad{t\in\mathcal{T}_{n}},\\
   &\quad\text{C}_{2}:\quad x_{n}(t+1)-\Xn\leq{x_{n}(t+r_{n}-1)},\quad{t\in\mathcal{T}_{n}},\\
    &\quad\text{C}_{3}:\quad 
    1\leq\sum\limits_{t\in\TK}x_{n}(t)\leq{r_{n}},
    \label{eq:agent's decision}
\end{aligned}
    \end{equation}
where $\mathbf{x_{n}}=\textbf{col}(x_{n}(1),\cdots,x_{n}(T))$. Constraints $\text{C}_{1}$ denotes that the maintenance decision is a binary variable and $\Xn=1$ indicates that agent $n$ decides to perform maintenance at time $t$. Constraints $\text{C}_{2}$ quantifies the amount of time needed for maintenance actions known as the repair time. Constraint $\text{C}_{3}$ enforces that agent $n$ performs maintenance in scheduling time horizon, which lasts $r_{n}$ time units at maximum.

Due to the stochastic nature of $\Second$, the optimization problem is also stochastic. To tackle it, we assume that $\Second$ has a finite number of possible realizations. Hence, we consider $\mathcal{S}_{n}={\{1,\cdots,S_{{n}}}\}$ scenarios for the second threshold $\Second$ and the probability of scenario $s$ is $\PI$, $s\in\mathcal{S}_{n}$, $\NN$. Hence, the optimization problem \eqref{eq:agent's decision} can be cast as a mixed integer linear programming problem:

\begin{equation}
\begin{aligned}
    \max\limits_{\mathbf{x_{n}}}&\sum\limits_{t\in\mathcal{T}_{n}}\sum\limits_{s\in\mathcal{S}_{n}}{\Big\{\big(1-\Xn\big)}\Big(\frac{\PI}{2}\big(1-\sign(t-\seconds)\big)\big(P({t})q_{n}({t})-C_{n}({t})\big)\\
    &-\alpha_{n}e^{(t-\first)}\big(\sign{(t-\first)}-{{\PI\sign(t-\seconds)}\big)\Big)}\Big\},\\
   \text{s.t.}&\quad\text{C}_{1},\quad\text{C}_{2},\quad\text{C}_{3}.
    \label{eq:agent's decision stochastic}
\end{aligned}
    \end{equation}

\subsection{Central system's objective function}
The central system  coordinator, i.e. the TSO, maximizes the system's reliability while fully supplying the power demand. The central system has a coordinating function. It cannot change the agents' decisions but only provide incentives to motivate them to change their decisions. This formulation corresponds to the setup in real applications where power generating agents are independent stakeholders and take their decisions independently of any central coordinating system. 

To this end, we propose a framework where the central system prioritizes the decisions of the agents whose plant is close to $\Second$, $\NN$, while ensuring the fulfillment of the energy demand by considering the network and agents constraints. Hence, we model the central system's objective as:

    \begin{equation}
\begin{aligned}
    \max\limits_{\mathbf{y},\mathbf{q},\mathbf{\theta},\mathbf{u}}&\quad\sum\limits_{n\in\mathcal{N}}\sum\limits_{t\in\mathcal{T}_{n}}\sum\limits_{s\in\mathcal{S}_{n}}\frac{\PI}{\seconds-t+\epsilon}\Yn,\\
    \text{s.t.}&\quad\text{A}_{1}:\quad\sum\limits_{n\in\mathcal{N}_{j}}(1-\Yn)q_{n}(t)-\sum\limits_{r\in\Gamma_{j}}B_{j,r}(\theta_{j}(t)-\theta_{r}(t))={L_{j}(t)},\quad{j\in\mathcal{J}},{t\in\mathcal{T}_{n}},\\
    &\quad\text{A}_{2}:\quad\Yn\in{\{0,1}\},\quad {t\in\mathcal{T}_{n}}, \NN,\\
    &\quad\text{A}_{3}:\quad \Yn\leq{\Xn},\quad {t\in\mathcal{T}_{n}},\NN\\
    &\quad\text{A}_{4}:\quad u_{n}(t)q_{n}^{\min}\leq{q}_{n}(t)\leq{u_{n}(t)q_{n}^{\max}},\quad {t\in\mathcal{T}_{n}}, \NN,\\
    &\quad\text{A}_{5}:\quad{u_{n}(t)}\in{\{0,1}\},\quad{t\in\mathcal{T}_{n}},\NN,\\
        &\quad\text{A}_{6}:\quad{-F_{j,r}}\leq{B}_{j,r}(\theta_{j}(t)-\theta_{r}(t))\leq{F_{j,r}},\quad r\in\Gamma_{j}, j\in\mathcal{J},{t\in\mathcal{T}_{n}},\\
        \label{eq:central system}
\end{aligned}
    \end{equation}
     where $\mathbf{y}=\textbf{col}(\mathbf{y_{1}},\cdots,\mathbf{y_{N}})$, and $\mathbf{y_{n}}=\textbf{col}(y_{n}(1),\cdots,y_{n}(T_{n}))$,  $\NN$.  $\mathbf{q}=\textbf{col}(\mathbf{q_{1}},\cdots,\mathbf{q_{N}})$, and $\mathbf{q_{n}}=\textbf{col}(q_{n}(1),\cdots,q_{n}(T_{n})), \NN$.
    $\mathbf{u}=\textbf{col}(\mathbf{u_{1}},\cdots,\mathbf{u_{N}})$, and $\mathbf{u_{n}}=\textbf{col}(u_{n}(1),\cdots,u_{n}(T_{n}))$, $\NN$. $\mathbf{\theta}=\textbf{col}(\mathbf{\theta_{1}},\cdots,\mathbf{\theta_{J}})$, and $\mathbf{\theta_{j}}=\textbf{col}(\theta_{j}(1),\cdots,\theta_{j}(T_{n}))$, $j\in\mathcal{J}$. 
    
    $\epsilon\geq{0}$ is a small constant value which avoids the infinity of the optimization and makes it computationally tractable.  
    Constraints $\text{A}_{1}$ ensures that power demand is satisfied at each bus (node) given the maintenance decisions of each unit. Constraints $\text{A}_{2}$ indicates that the central system's decision is a binary variable. Constraints $\text{A}_3$ enforces that the central system just accepts and rejects the agents' decisions and does not force them to perform maintenance. Constraints $\text{A}_{4}$ denotes the power generation limit of each agent. Constraints $\text{A}_{5}$ indicates that the unit commitment decision variable is binary, i.e. it is equal to one if the agent is scheduled to be committed in each period of time and is zero, otherwise. Constraints $\text{A}_{6}$ ensures that the flow in the transmission lines is within the capacity limits. Market clearing is performed using the maximum daily demand. Therefore, we can assume that the units will be able to ramp up/down to the required power level within the time frame of one day.
    
    In order to make $\Yn{q_{n}(t)}$ linear, we substitute it with a new variable:
    
    \begin{equation}
        Z_{n}(t)=(1-\Yn){q_{n}(t)},
        \label{eq:z}
    \end{equation}
    and add the following constraints as the equivalent of the nonlinear term
  
    \begin{equation}
    \begin{aligned}
        &{0}\leq Z_{n}(t)\leq(1-y_{n}(t))M\\
        &{Z}_{n}(t)\geq{q_{n}(t)-(1-(1-\Yn))M}\\
        &Z_{n}(t)\leq{q_{n}(t)+(1-(1-\Yn))M},
        \label{eq:big m}
        \end{aligned}
    \end{equation}
    where $M$ is a large positive constant \cite{Fortuny_1981}.
    
   Substituting \eqref{eq:z} and \eqref{eq:big m} into \eqref{eq:central system}, we obtain a mixed integer linear programming optimization:
    
    \begin{equation}
\begin{aligned}
    \max\limits_{\mathbf{y},\mathbf{q},\mathbf{\theta},\mathbf{u},\mathbf{Z}}&\quad\sum\limits_{n\in\mathcal{N}}\sum\limits_{t\in\mathcal{T}_{n}}\sum\limits_{s\in\mathcal{S}_{n}}\frac{\PI}{\seconds-t+\epsilon}\Yn,\\
    \text{s.t.}
    &\quad\text{A}_{2},\quad\text{A}_{3},\quad\text{A}_{4},\quad\text{A}_{5}, \quad\text{A}_{6},\\
    &\quad\text{A}_{7}:\quad\sum\limits_{n\in\mathcal{N}_{i}}Z_{n}(t)-\sum\limits_{r\in\Gamma_{j}}B_{j,r}(\theta_{j}(t)-\theta_{r}(t))={L_{j}(t)},\quad{j\in\mathcal{J}},{t\in\mathcal{T}_{n}},\\
        &\quad\text{A}_{8}:\quad{0}\leq Z_{n}(t)\leq(1-y_{n}(t))M,\quad {t\in{\{1,\cdots,T}\}},\NN,\\
         &\quad\text{A}_{9}:\quad{Z}_{n}(t)\geq{q_{n}(t)-(1-(1-\Yn))M},\quad {t\in\mathcal{T}_{n}},\NN,\\
         &\quad\text{A}_{10}:\quad Z_{n}(t)\leq{q_{n}(t)+(1-(1-\Yn))M},\quad {t\in\mathcal{T}_{n}},\NN,
        \label{eq:central system1}
\end{aligned}
    \end{equation}
where $\mathbf{Z}=\textbf{col}(\mathbf{Z_{1}},\cdots,\mathbf{Z_{N}})$ and $\mathbf{Z_{n}}=\textbf{col}(Z_{n}(1),\cdots,Z_{n}(T_{n}))$, $\NN$.

\section{Coordination procedure via incentive signal}
\label{sec:coordinated}

For the agent $n$, time slots for the maintenance action have to be identified, such that the TSO accepts the maintenance decision proposed by agent $n$. The TSO accepts the maintenance decision, if the overall system power demand is fully supplied and at the same time the agent can perform the required maintenance actions on the generating units. Since the objectives of the TSO and the generating units are partly conflicting, we solve the conflict by introducing a negotiation process between the agents and the central system using an incentive signal. 

In the first iteration $i=1$ of the negotiation algorithm, the agents set their decisions based on Equation \eqref{eq:agent's decision stochastic}, and submit them to the central system. The central system solves Equation \eqref{eq:central system1} for making a decision about whether to accept or reject the agents' maintenance decisions. If ${\Yn=1}$ and ${\Xn=1}$, agent $n$ can perform maintenance, otherwise the decision of agent $n$ is rejected and ${\Yn=0}$. In the case that agents' decisions cannot be accepted because the power demand cannot be satisfied, the central system sends an incentive signal to the agents which affects their objective function and motivates them to change their maintenance decisions in the next iteration of the negotiation algorithm. Without imposing an incentive signal, the agents whose decisions have not been accepted by the central system cannot perform maintenance before failure and would fail. Hence, they would not be able to produce power, which is also not in the interest of the central system. Therefore, an incentive signal is created such that the maintenance decisions of all agents will be accepted by the central system at some point (iteration).  In other words, the incentive signal is the contract between central system and agents, in the sense that if the agents' decisions cannot be accepted by the central system they have to pay a penalty to TSO, unless they change their decisions in the next iterations until their decisions can be accepted by TSO.

We define the incentive signal for agent $n$ at iteration $i$ and time $t$ as:

\begin{equation}
\begin{aligned}
    &I^{i}_{n}(t)=\gamma_{n}x^{i}_{n}(t)\sign\Big(\sum_{o=1}^{i-1}\big({y_{n}^{o}(t)-x_{n}^{o}(t)\big)\Big)},
           \label{eq:incentive}
\end{aligned}
\end{equation}
where $\X$ indicates the decision of agent \textit{n} at iteration \textit{i} and time $t$, and $y_{n}^{o}(t)$ indicates the central system decision at iteration ${o}$. The rationale of the incentive signal is detailed in Section \ref{subsec:rationale}.

Hence, during the negotiation process at iteration $i$, the objective function and the optimization problem for agent $n$ are expressed as:

\begin{equation}
\begin{aligned}
    &\max\limits_{\mathbf{x_{n}^{i}}}\sum\limits_{t\in\mathcal{T}_{n}}\sum\limits_{s\in\mathcal{S}_{n}}{\Big\{\big(1-x_{n}^{i}(t)\big)}\Big(\frac{\PI}{2}\big(1-\sign(t-\seconds)\big)\big(P({t})q_{n}({t})-C_{n}({t})\big)\\
    &-\alpha_{n}e^{(t-\first)}\big(\sign{(t-\first)}-{{\PI\sign(t-\seconds)}\big)\Big)}\Big\}+\sum\limits_{t\in\mathcal{T}_{n}}I^{i}_{n}(t),\\
  &\;\text{s.t.}\quad\text{C}^{\prime}_{1}:\quad{x}_{n}^{i}(t)\in{\{0,1}\},\quad{t\in\mathcal{T}_{n}},\\
   &\quad\;\;\quad\text{C}^{\prime}_{2}:\quad x_{n}^{i}(t+1)-x_{n}^{i}(t)\leq{x_{n}^{i}(t+r_{n}-1)},\quad{t\in\mathcal{T}_{n}},\\
    &\quad\;\;\quad\text{C}^{\prime}_{3}:\quad 
    1\leq\sum\limits_{t\in\TK}x_{n}^{i}(t)\leq{r_{n}},
       \label{eq:agent's decision1}
\end{aligned}
    \end{equation}
    where $\mathbf{x_{n}^{i}}=\textbf{col}(x_{n}^{i}(1),\cdots,x_{n}^{i}(T_{n}))$.

    The negotiation algorithm for the maintenance decision is described in Algorithm \ref{alg:iterative algorithm}.

\begin{algorithm}
	\caption{Negotiation algorithm for maintenance decisions} 
	\label{alg:iterative algorithm}
	\begin{algorithmic}[1]
		\STATE {\textbf{Input}}: $x_{n}^{0}(t)=0$,\quad $y_{n}^{0}(t)=0$, $\NN$, $t\in\mathcal{T}_{n}$.\\       
		\textbf{Iterate:}   \\                                             \label{alg.l.1}
		\STATE\text{For $\NN$} repeat until convergence:
		\STATE Obtain $x_{n}^{i}(t)$ using Equation \eqref{eq:agent's decision1}, $t\in\mathcal{T}_{n}$.
		\STATE Obtain $y_{n}^{i}(t)$ using Equation \eqref{eq:central system1}, $t\in\mathcal{T}_{n}$. 
		\STATE If $y_{n}^{i}(t)\neq{x_{n}^{i}(t)}$, $t\in\mathcal{T}_{n}$, calculate incentive signal using Equation \eqref{eq:incentive}. 
		\STATE  $i\leftarrow{i+1}$.
	\end{algorithmic}
\end{algorithm}

\subsection{Rationale of the proposed incentive signal}
\label{subsec:rationale}

If the maintenance decision of agent $n$ cannot be accepted by the central system at iteration $i$ and time $t$, the agent will receive $-\gamma_{n}\Xni$ as a penalty function in every iteration after $i$. Since the agent seeks to maximize its objective function, it will likely change its maintenance decision. Additionally, if the agents' decision is accepted by the central system at time $t$ in iteration $i$, then the incentive signal is zero at time $t$ in iteration $i+1$. Hence, if the agents obtain their maximum reward at iteration $i$ by choosing the maintenance decision at time $t$, the same decision is also made in iteration $i+1$. Namely, agents are not forced to make the same decision in subsequent iterations. However, the agents choose the same decisions through the optimization because they are optimal for them.

The developed coordination framework embraces the following assumptions.

\begin{assumption}
\label{eq:assumption1}
The repair time $r_{n},\NN$ is sufficiently small with respect to the maintenance scheduling time horizon. In the worst case where $\first$ and $\seconds, s\in\mathcal{S}_{n},$ are equal for all the agents:
\begin{equation}
  {\seconds}-\first\geq{N\max\limits_{\NN}} ({r_{n}}), \quad s\in\mathcal{S}_{n}, \NN.  
\end{equation}
\end{assumption}

\begin{assumption}
\label{eq:assumption2}
If one agent can perform maintenance, the power demand is satisfied by the remaining operating agents. In other words, if agent $m\neq{n}$, $m\in\mathcal{N}$, $\NN$, decides to perform maintenance at time $t\in\mathcal{T}_{n}$, we have 
\begin{equation}
\sum\limits_{\NN,n\neq{m}}q_{n}^{\max}\geq{L(t)}.
\end{equation}
\end{assumption}

Assumption \ref{eq:assumption1} and \ref{eq:assumption2} 
are not too conservative in real power systems. Since agents are heterogeneous, the worst-case that $\first$ and ${\seconds}$, $s\in\mathcal{S}_{n}, \NN$ are equal for all the agents is not realistic due to the individual variability of the operating conditions and differences in system configurations. Moreover, the failure rates are typically small due to the high requirements on system availability and safety. Therefore, it is reasonable to assume that the repair time is small relative to ${\seconds}-\first$. Furthermore, in most of the real electricity markets there is sufficient overcapacity. Hence, when just one agent performs maintenance the power demand can be satisfied.

The following lemma imposes a condition on $\gamma_{n}$ which is the constant for each agent and depends on the reward that the agent can obtain by its power generation. This condition guarantees that when the agents' decisions cannot be accepted by the central system, the agents will change their decisions at the next iteration.

\begin{lemma}
\label{eq:lemma}
Consider $x_{n}^{i}(t)$ and $y_{n}^{i}(t)$ as the maintenance decision of agent $n$ and the central system at time $t$ and iteration $i$ of the algorithm. In the case that $x_{n}^{i}(t)$ is not equivalent to $y_{n}^{i}(t)$,  $x_{n}^{i+1}(t)$ will not be equal to $x_{n}^{i}(t)$ if we have:
\begin{equation}
\begin{aligned}
\gamma_{n}&\geq{\max{\Big(0,\max_{t\in\mathcal{T}_{n}}\sum\limits_{s\in\mathcal{S}_{n}}{\Big\{}\Big(-\frac{\PI}{2}\big(1-\sign(t-\seconds)\big)\big(P({t})q_{n}({t})-C_{n}({t})\big)}}\\
    &+\alpha_{n}e^{(t-\first)}\big(\sign{(t-\first)}-{{\PI\sign(t-\seconds)}\big)\Big)}\Big\}\Big),\quad\NN.\\
\label{eq:gamma1}
\end{aligned}
\end{equation}
\end{lemma}

\begin{proof}
If $x_{n}^{i}(t)$ is not equal to $y_{n}^{i}(t)$, it means that $x_{n}^{i}(t)$ equals to one, hence the agent $n$ gets $-\gamma_{n}$ as the penalty function. We can conclude that a rational agent $n$ will not choose $x_{n}^{i+1}(t)=1$ if it gets a lower reward than $x_{n}^{i+1}(t)=0$ (a different decision would not be rational). Hence, this results in:
\begin{equation}
\begin{aligned}
-\gamma_{n}\leq{\max_{t\in\mathcal{T}_{n}}\sum\limits_{s\in\mathcal{S}_{n}}{\Big\{}\Big(\frac{\PI}{2}\big(1-\sign(t-\seconds)\big)\big(P({t})q_{n}({t})-C_{n}({t})\big)}\\
    -\alpha_{n}e^{(t-\first)}\big(\sign{(t-\first)}-{{\PI\sign(t-\seconds)}\big)\Big)}.
\label{eq:gamma2}
\end{aligned}
\end{equation}
We ensure that Equation \eqref{eq:gamma2} holds if we have

\begin{equation}
\begin{aligned}
\gamma_{n}&\geq{\max_{t\in\mathcal{T}_{n}}\sum\limits_{s\in\mathcal{S}_{n}}{\Big\{}\Big(-\frac{\PI}{2}\big(1-\sign(t-\seconds)\big)\big(P({t})q_{n}({t})-C_{n}({t})}\\
    &+\alpha_{n}e^{(t-\first)}\big(\sign{(t-\first)}-{{\PI\sign(t-\seconds)}\big)\Big)}\Big\},\quad\NN,
\label{eq:gamma3}
\end{aligned}
\end{equation}
since $\gamma$ should be positive, so \eqref{eq:gamma1} must be satisfied.
\end{proof}

\begin{theorem}
\label{eq:theory}
Algorithm \ref{alg:iterative algorithm} converges, in the sense that
\begin{equation}
\lim_{i\to\infty}y_{n}^{i}(t)=x_{n}^{i}(t),\quad t\in\mathcal{T}_{n},\quad \NN.
\end{equation}
\end{theorem}

\begin{proof}
Lemma $1$ ensures that agents whose decisions cannot be accepted by the central system will change their decisions at the next iteration. In addition, Assumptions $1$ and $2$ ensure that there is sufficient time between ${\seconds}$ and $\first$, $s\in\mathcal{S}_{n}$, $\NN$, such that all the agents' decisions can be accepted by the central system at the end of the negotiation process.
\end{proof}

\subsection{Budget balance}

\begin{definition}(Budget balance)
The negotiation mechanism is budget balanced if the cumulative amount of incentives/penalties that the central system has to pay or receive from the agents in every iteration of the algorithm would be equal to zero \cite{Herriges_1992}. This condition is expressed by:
\begin{equation}
 \SUMN\sum\limits_{t\in\mathcal{T}_{n}}{I}^{i}_{n}(t)=0.
\label{eq:BUDEGT}
\end{equation}
\end{definition}

\begin{definition}(Weak budget balance)
The negotiation mechanism is weak budget balanced if the cumulative amount of penalties that the central system gets from the agents is larger than the cumulative amount of incentives that it pays to the agents. In other words, the revenue of the central system is positive \cite{Krishna_2009}. It can be expressed as follows:
\begin{equation}
 \SUMN\sum\limits_{t\in\mathcal{T}_{n}}{I}^{i}_{n}(t)\leq{0}.
\label{eq:weak BUDEGT}
\end{equation}
\end{definition}

In the proposed algorithm, at each iteration $\YT\leq\XT$ and, therefore, ${I}^{i}_{n}(t)\leq{0}$, $\NN$, ${t\in\mathcal{T}_{n}}$. Hence, we can deduce that the TSO does not need to pay the money to the agents to change their decisions and if the agents want to stop at any iteration before the convergence occurs they have to pay the money to TSO, hence our mechanism has a weak budget balance property. 

Moreover, if $\XT$ cannot be accepted by the central system at iteration $i$, then, agent $n$ receives the incentive signal $I_{n}^{i+1}(t)=-\gamma_{n}x_{n}^{i+1}(t)$ at iteration $i+1$, $\NN$, ${t\in\mathcal{T}_{n}}$. Using Lemma $1$, we make sure that $x_{n}^{i+1}(t)=0$, therefore, $I_{n}^{i+1}(t)=0$, $\NN$, ${t\in\mathcal{T}_{n}}$. Furthermore, since every agent's decisions are accepted at the convergence point, denoted by $i^*$, using Theorem $1$, we can conclude that

\begin{equation}
 \SUMN\sum\limits_{{t\in\mathcal{T}_{n}}}{I}^{i^*}_{n}(t)={0}.
\label{eq:weak BUDEGT1}
\end{equation}

Hence, our mechanism is budget balance at convergence. 
\section{Case study}
\label{sec:case}

The developed algorithm is applied to the IEEE 39 bus New England system \cite{Bills_1970}. The system consists of 39 buses, 29 lines, 46 branches of which 12 transformers, and 10 generating units with the total generating capacity of 7367 MW. The yearly load curve  with peak demand of 6254 MW and minimum demand of 3026 MW is used (Figure \ref{fig:demand}).

\begin{figure}[H]
\centering
\includegraphics[scale=0.8]{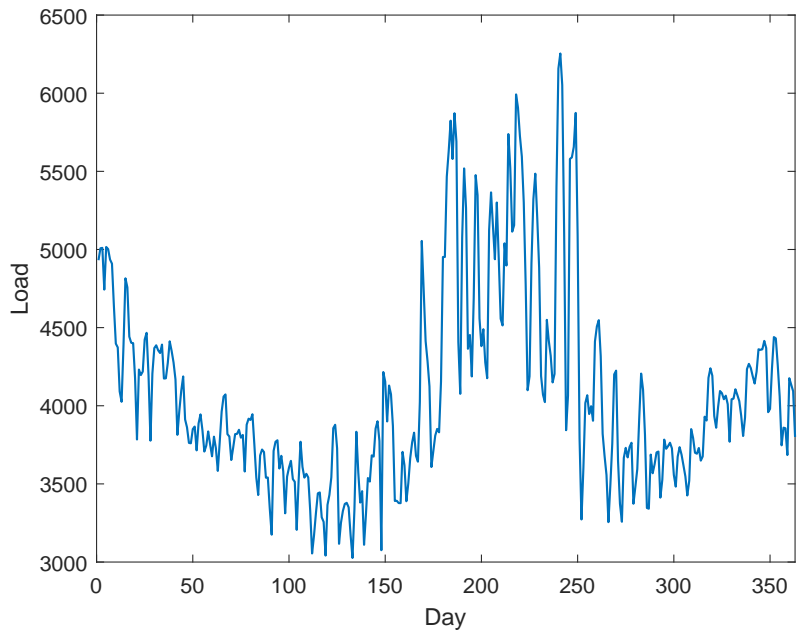}
\caption{A yearly demand curve for the IEEE 39 bus system.}
  \label{fig:demand}
\end{figure}

\section{Results}
\label{sec:simulation}

The results of the maintenance algorithm are shown for a period of $365$ days which consists of different maintenance scheduling time horizon with $50$ scenario paths $(S_{n}=50)$, $\NN$, for the remaining useful life which is a random variable. In particular, the uncertainty of the RUL prediction for agent $n$ is assumed to obey a normal distribution with mean $\tau_{n}$ and standard deviation $\sigma_{n}$. The values of the first threshold $t_{1,n}$ and the parameters of the normal distribution for second threshold $t_{2,n}$ for different maintenance scheduling time horizons for all the agents are displayed in Table \ref{tab:thereshold}. 

\begin{table}
  \caption{The first threshold and second threshold distribution parameters}
  \centering
 \begin{tabular}{|c  c c c|} 
 \hline
 Agent &  $t_{1,n}$ & $\tau_{n}$ & $\sigma_{n}$ \\ [0.5ex]
 \hline
 1 & [55, 172, 290] & [115, 251, 380] & [6.67, 6.56, 5.45] \\ 
 \hline
 2 &  [31, 132, 226, 282] & [91, 185, 258,360] & [5.06, 6.92, 5.77, 5.15] \\
 \hline
 3 & [48, 152, 266, 363] & [109, 201, 306, 396] & [6.32, 6.95, 5.89, 6.23]\\
 \hline
 4 & [22, 125, 200, 272] & [82, 165, 241, 320, 400] & [5.74, 5.44, 6.27, 6.39, 5.25]\\
 \hline
 5 & [24, 115, 186, 273] & [84, 161, 239, 324, 416]& [6.28, 5.20, 5.03, 5.32, 6.05] \\
 \hline
 6 & [54, 173, 283] & [114, 236, 371] &[5.14, 6.70, 6.39]\\
 \hline
 7 & [37, 159, 217, 301] & [97, 178, 272, 365]& [5.48, 6.09, 6.32, 6.88] \\
 \hline
 8 & [42, 125, 244, 363] & [102, 212, 330, 435]&[5.59, 5.63, 6.33, 5.23]\\
 \hline
 9 & [49, 157, 264, 363] & [109, 224, 335, 450]&[5.88, 6.53, 6.81, 6.76]\\
 \hline
 10 & [58, 173, 304, 417] & [118, 247, 349, 457]& [5.61, 6.59, 5.81, 5.77]\\
 \hline
\end{tabular}
  \label{tab:thereshold}
\end{table}

\begin{remark}
The selected values of $t_{1,n}$ and $\tau_{n}$, $\NN$, shown in Table \ref{tab:thereshold} exemplify relatively large failure rates for each agent. This selection allows demonstrating the negotiation process in the simulation results for the size of the test system. Nonetheless, the proposed coordination algorithm can account for planning decision horizons of several years.
\end{remark}

The convergence of the decision making process of the agents and central system at each iteration of the algorithm are shown in Figure \ref{fig:maintenance decision}.

\begin{figure}[H]
\centering
{\onecolumn\includegraphics[scale=0.6]{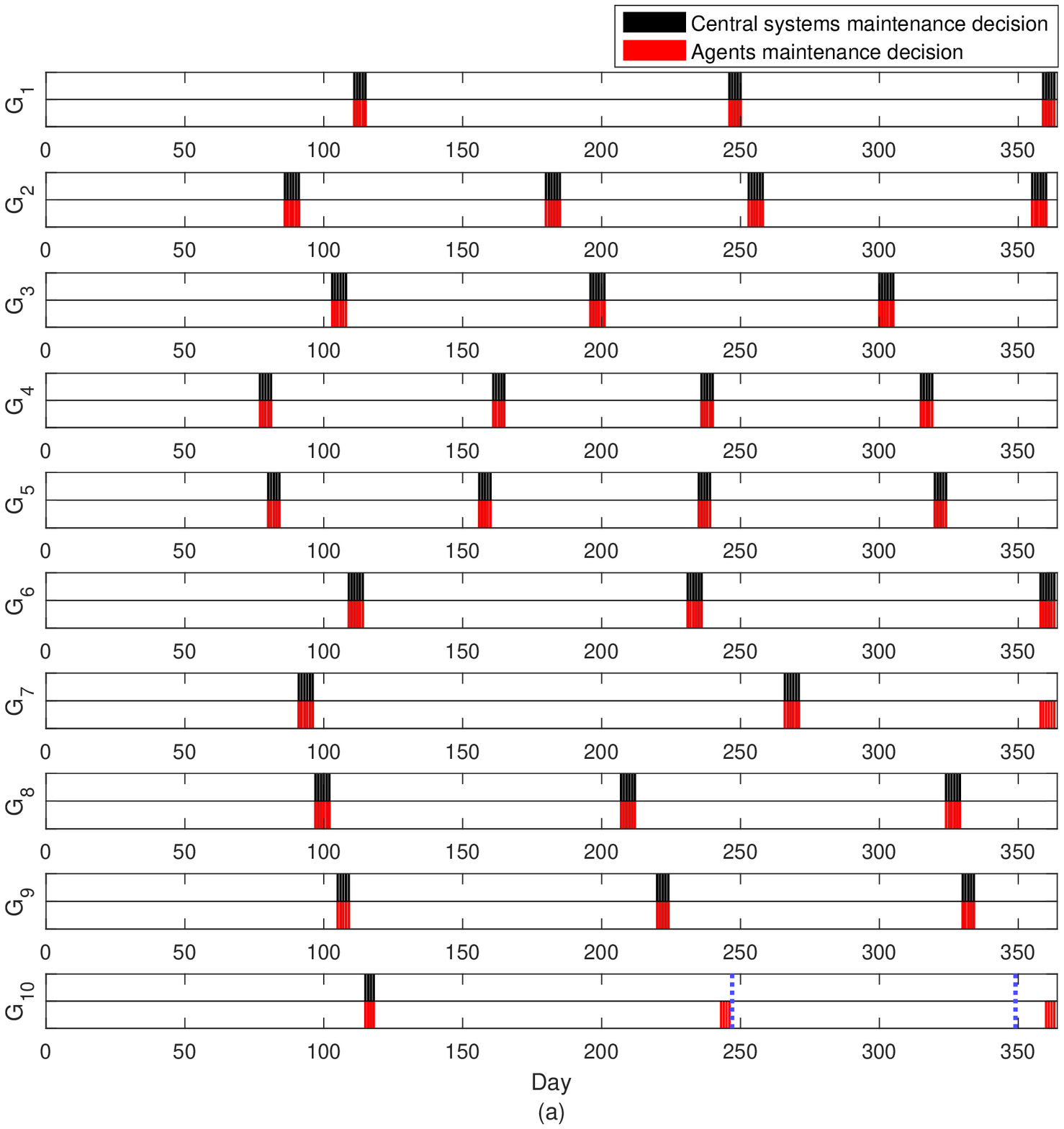}}\\
{\onecolumn\includegraphics[scale=0.6]{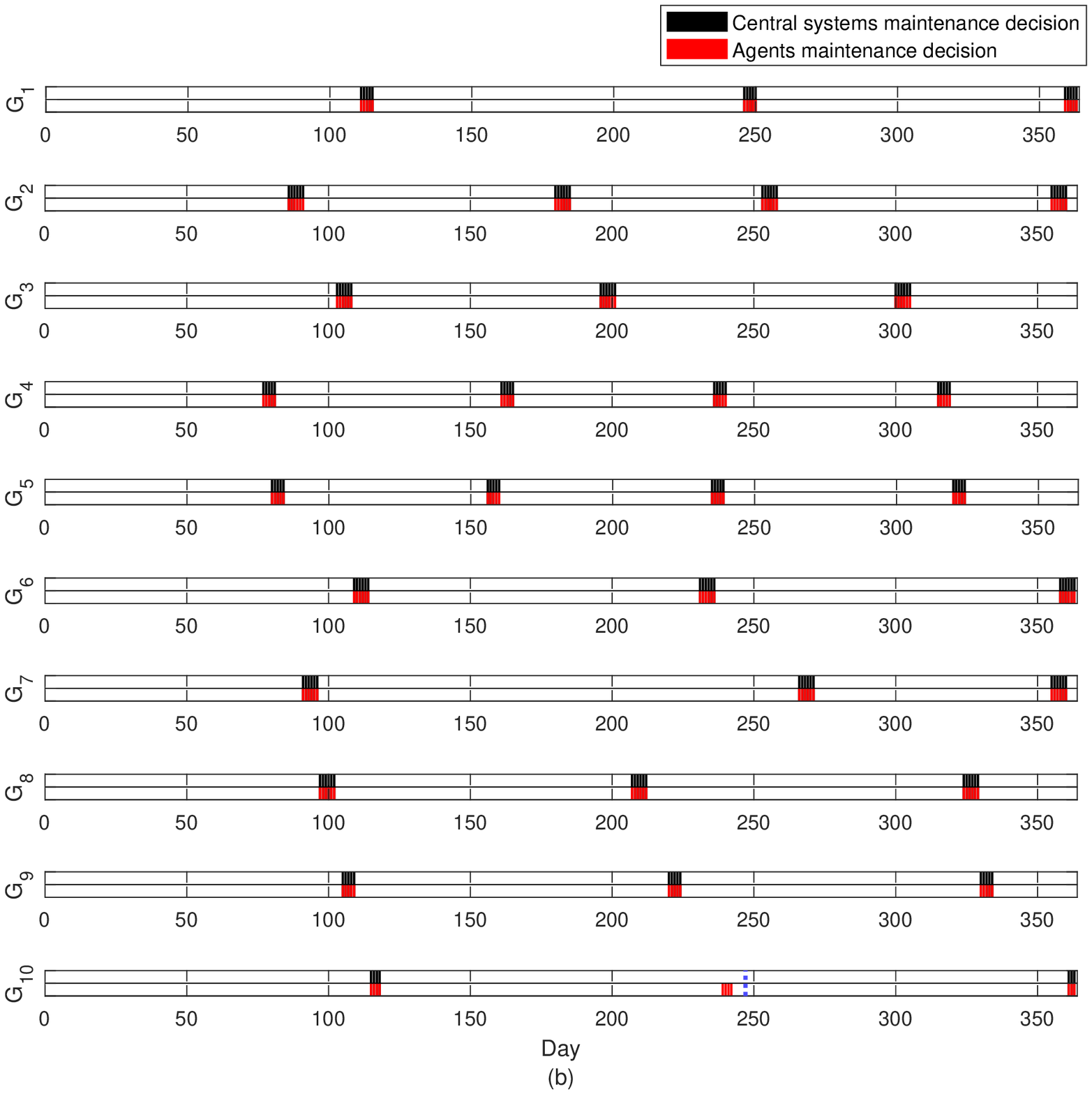}}
 \end{figure}

\begin{figure}[H]
\centering
{\onecolumn\includegraphics[scale=0.6]{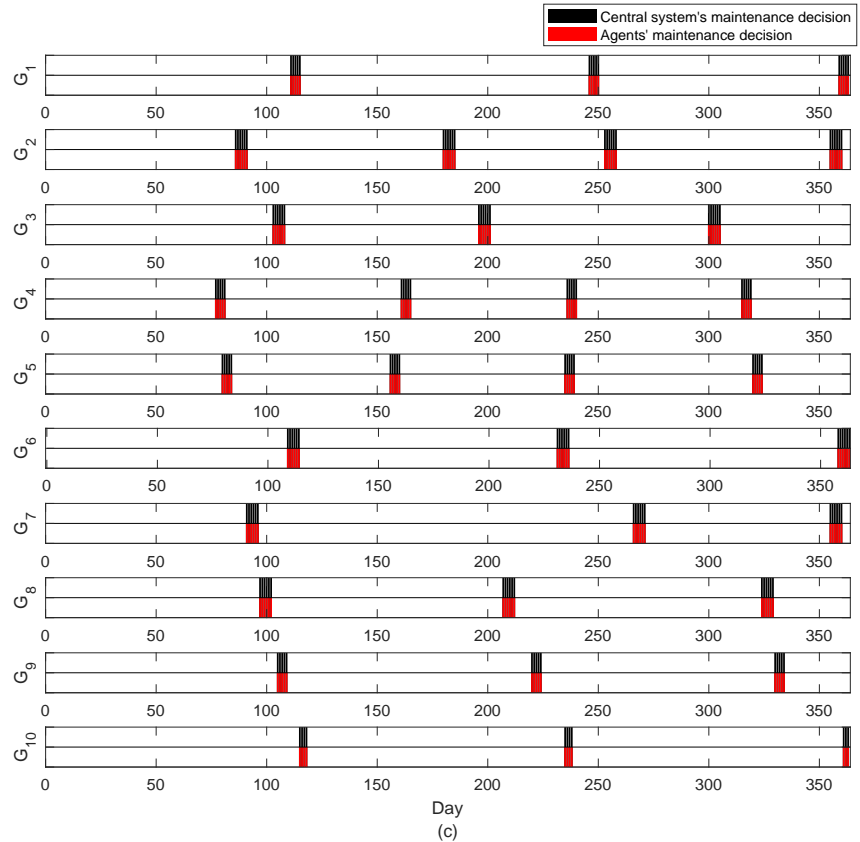}}
  \caption{Maintenance decisions of each of the 10 agents and the central system (a) at the first iteration (b) at the second iteration (c) at the third (final) iteration when the algorithm converges. The dashed vertical lines for agent 10 in panel (a) and (b) indicate the mean value $\tau_{10}$ of the failure time distribution for agent 10.}
  \label{fig:maintenance decision}
  \end{figure}
Figure \ref{fig:maintenance decision} shows that the decisions of agents $7$ and $10$ are not acceptable for the central system in the first iteration of the algorithm (Figure \ref{fig:maintenance decision}(a)). In this case, the two generating units would fail if they do not perform maintenance in each maintenance scheduling time horizon. Therefore, we introduce an incentive signal which encourages the agents to change their decisions. Indeed, in the next iteration, agents $7$ and $10$ change their decisions, such that the central system accepts the decision of agent $7$ (Figure \ref{fig:maintenance decision}(b)) rejects the decision of agent $10$. The negotiation is, therefore, repeated until the convergence is reached. Finally, all the agent's decisions are accepted by the central system in the last iteration as shown in Figure \ref{fig:maintenance decision}(c). Therefore, we can conclude that the algorithm converges and after all decisions are accepted the agents' decisions will not change. Figure \ref{fig:capacity} shows the maximum available capacity during the iterations of the negotiation process.

\begin{figure}[H]
\centering
\includegraphics[scale=0.5]{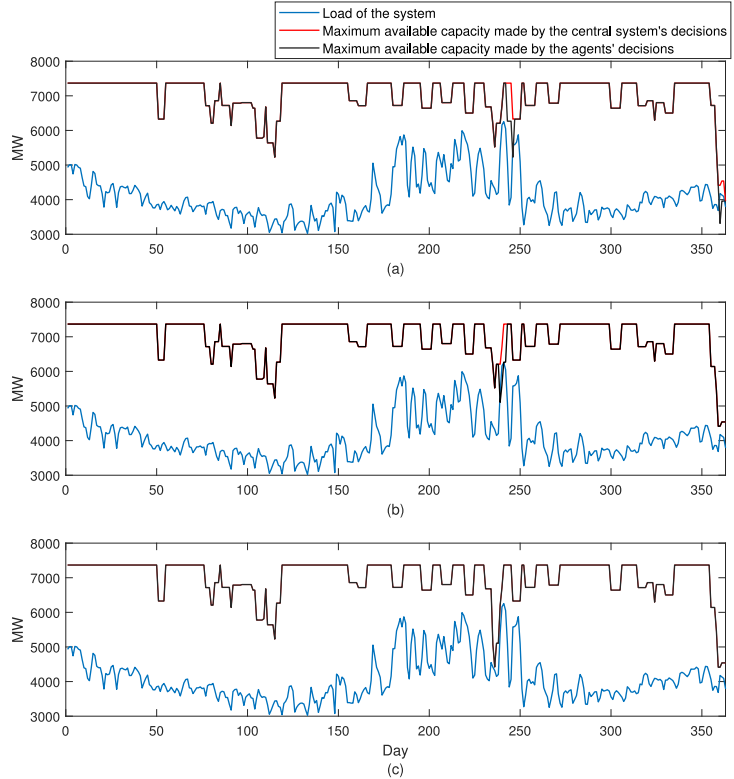}
\caption{Fulfilled capacity during the negotiation process (a) at the first iteration (b) at the second iteration (c) at the third iteration.}
  \label{fig:capacity}
\end{figure}
Figure \ref{fig:capacity} shows that during the first and second iterations of the algorithm the available capacity that agents can fulfill is less than the load of the system. Hence, the central system cannot accept the maintenance decisions of all the agents. At the third iteration, the available capacity made by the agents' decisions is higher than the power demand. Therefore, all the agents' decisions are accepted.

We compare the incentives or penalties of agents whose decisions could not be accepted in at least one iteration (agents $7$ and $10$) in Table \ref{tab:incentive}.

\begin{table}[H]
  \caption{Incentive signals for agents [\$]}
\centering
 \begin{tabular}{|c c c c|}
 \hline
 Agent's number & Iteration $1$ & Iteration $2$ & Iteration $3$ \\ [0.5ex] 
 \hline
 7 & -200,000 & 0 & 0 \\
 \hline
 10 & -800,000 & -400,000 & 0  \\ 
 \hline
\end{tabular}
  \label{tab:incentive}
\end{table}
Table \ref{tab:incentive} shows that agents $7$ and $10$ receive the penalty function in iteration $1$ because their maintenance decisions could not be accepted. Then, in the second iteration, agent $7$ modifies its decisions and does not need to pay the penalty to the central system and its decisions can be accepted. However, in the second iteration agent $10$ receives a penalty function since its decision is rejected by the central system. At iteration $3$, the algorithm converges and agents receive no penalties anymore. Hence, based on results reported in Table \ref{tab:incentive}, the revenue of the central system during the negotiation is positive. Therefore, the mechanism is weak budget balanced. At iteration $3$ (convergence point), the summation of the incentive signals is zero and the mechanism is budget balanced.

We compare the change of achieved rewards for agents $7$ and $10$ in each iteration of the algorithm during the negotiation process in Table \ref{tab:reward}.

\begin{table}[H]
  \caption{Expected rewards of the agents by considering the incentive function during negotiation [\$]}
  \centering
 \begin{tabular}{|c c c c|} 
 \hline
 Agent's number & Iteration $1$ & Iteration $2$ & Iteration $3$ \\ [0.5ex] 
 \hline
 7 & 503,800 & 704,100 & 704,100   \\ 
 \hline
 10 & 662,900 & 1,054,000 & 1,455,800 \\
  \hline
\end{tabular}
  \label{tab:reward}
\end{table}
 Table \ref{tab:reward} shows that the rewards of the agents increase during the negotiation since in each iteration after they receive the incentive (penalty) signal, they change their decisions such that the penalty signal decreases in the next iteration. Moreover, the rewards of all these agents are positive, even when they get a negative incentive signal. We can infer from these positive rewards that the agents will participate in the negotiation algorithm voluntarily, assuming that the agents are rational, since they get positive rewards instead of no rewards when they do not participate in the negotiation. Hence, our proposed incentive signal scheme has the individual rationality feature. 

\subsection{Comparison of the proposed algorithm to the baseline decisions}
The results of the proposed algorithm are compared against the baseline solutions where agents decide about their maintenance in each scheduling time at $t_{1,n}$, $\NN$. This corresponds to the condition-based maintenance scenario where failures are detected but no remaining useful lifetime prediction can be provided, $t_{2,n}$, $\NN$ are, therefore, not known. Moreover, we compare the results of the proposed algorithm  with the maintenance decision at $\tau_{n}$, $\NN$ where the units have high failure probability. This scenario corresponds to the scenario of corrective maintenance. While both scenarios (condition-based and corrective) are displayed in one figure, it is important to note that the maintenance is performed only once either at $t_{1,n}$, or at $\tau_{n}$, $\NN$.

The decisions of agents at $t_{1,n}$ and $\tau_{n}$, $\NN$, are shown in Figure \ref{fig:baselinet1}.

\begin{figure}[H]
\centering
\includegraphics[scale=0.6]{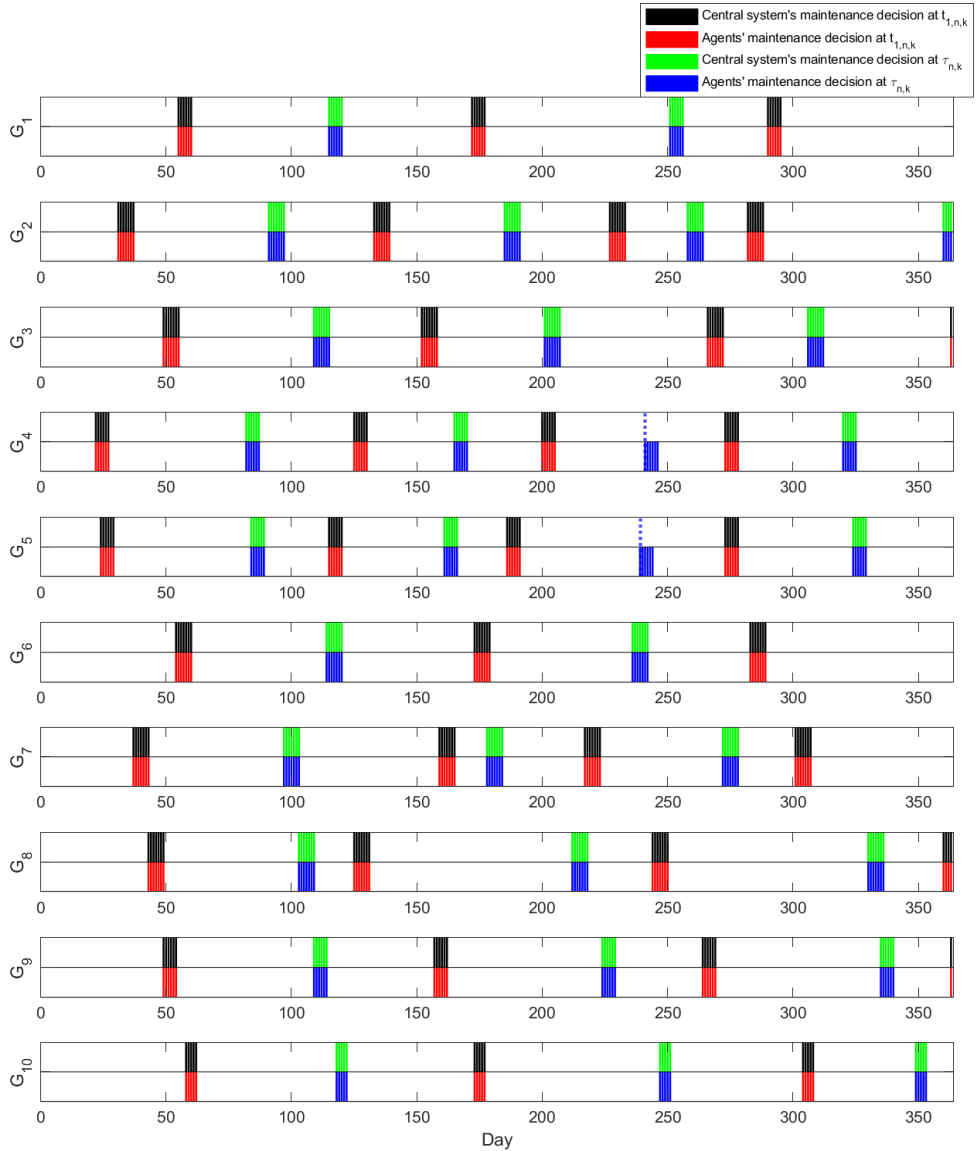}
\caption{Maintenance decisions at $t_{1,n}$ (condition-based maintenance scenario) and $\tau_{n}$, $\NN$ (corrective maintenance scenario), (the decisions of agents 4 and 5 at $\tau_{n}$ would not fulfill the load. The
dashed vertical lines for agents 4 and 5 indicate the mean values $\tau_{4}$ and $\tau_{5}$ of the failure time distributions for agents 4 and 5.)}
  \label{fig:baselinet1}
\end{figure}

We compare the rewards of the agents for the three analysed cases: decisions at $t_{1,n}$ (condition-based), $\tau_{n}$, $\NN$ (corrective), and based on the proposed negotiation algorithm (Figure \ref{fig:base line cost}). 
\begin{figure}[H]
\centering
\includegraphics[scale=0.6]{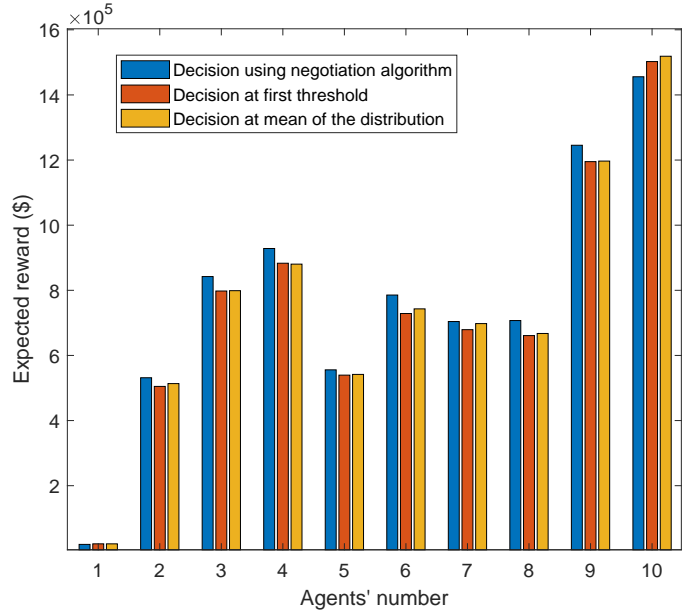}
\caption{Comparison of the rewards in the three cases: Negotiation algorithm, condition based decision $t_{1,n}$, decision at $\tau_{n}$, $\NN$.}
  \label{fig:base line cost}
\end{figure}
Figure \ref{fig:base line cost} shows that the agents' rewards with the proposed negotiation algorithm are larger than or equal to the rewards in the other two cases, (assuming that their decisions are accepted by the central system). This outcome is expected because the agents' decisions in the negotiation algorithm are based on the optimization approach which maximizes their rewards (Equation \eqref{eq:agent's decision stochastic}). However, when the agents decisions cannot be accepted by the central system and the agents receive an incentive signal, their decisions do not maximize the reward functions (Equation \eqref{eq:agent's decision stochastic}). Therefore, the reward could be smaller than the rewards that they would obtain by the decision at $t_{1,n}$ or $\tau_{n}$, $\NN$. Agent $10$ is the agent who gets the incentive signal and, hence, its rewards is smaller than the reward for the decisions at $t_{1,n}$ and $\tau_{n}$, $\NN$. 

It is important to point out that the baseline decisions are not an efficient method to obtain the maintenance scheduling since some of the agents' decisions cannot always be accepted by the central system as shown in Figure \ref{fig:baselinet1}, i.e. the maintenance decisions of agents' $4$ and $5$ cannot be accepted by the central system. In this case, they would fail and would not get any reward.  Please note that the penalties that would be imposed for the case that the generating units are failing and the demand would not be able to be fulfilled are not explicitly considered in this research. This would be the case for units 4 and 5 in Figure \ref{fig:base line cost}: the central system would not accept those maintenance actions since the demand would not be able to be fulfilled. However, since it is the corrective maintenance scenario, the generating units do not have a choice when to perform maintenance. The maintenance needs to be performed after the failure occurred. 

\subsection{Impact of the number of scenarios $S_{n}$ and of the standard deviation $\sigma_{n}$ of RUL}
In this subsection, we study the effect of the number of scenarios $S_{n}$ and of the standard deviation, $\NN$ of the uncertainties associated with the RUL predictions, on the total expected reward of all agents. For the standard deviation $\sigma_{n}$, $\NN$, of the normal distribution of the uncertainty for the RUL predictions, we consider three ranges, i.e. $[1; 3]$, $[5; 7]$ and $[10; 12]$, which represent low, medium and high values. The resulting total expected rewards are $8.83e+06$ $\$$,  $8.80e+06$ $\$$ and $8.74e+06$ $\$$, respectively, when we have $50$ scenarios. Hence, the expected reward decreases about $2\%$ value range as the standard deviation increases. In other words, the increasing uncertainty
of the predictions, i.e. the increasing standard deviation, decreases the expected reward.

Furthermore to analyze the impact of the number of scenarios $S_{n}$, we consider the medium range for the standard deviation and $10$, $50$ and $100$ scenarios. The resulting total expected rewards are $8.74e+06$ $\$$, $8.80e+06$ $\$$ and $8.80e+06$ $\$$, respectively. When the number of scenarios increases, the expected reward increases about $1\%$ value range and saturates at approximately 50 scenarios, after which the agents’ rewards do not increase significantly. Hence, increasing the number of scenarios of RUL prediction improves the uncertainty description and increases the expected rewards. This effect is consistent with the effect of reducing $\sigma_{n}$ on the expected reward. Indeed, a minimum number of scenarios is required to achieve a sufficiently accurate representation of the RUL prediction distribution. However, when the number of scenarios increases, the computational time also increases. As an example, using 10 and 100 scenarios entails a computation time for each iteration of $330$ s and $3310$ s, respectively. 

\section{Conclusions}
\label{sec:conclusion}
We propose a novel bi-level negotiation framework to solve the generation maintenance scheduling problem in the context of predictive maintenance where remaining useful lifetime (RUL) can be predicted with an estimated uncertainty. Within this framework, we propose a model for the agents' objective function which allows the maintenance actions to be obtained based on the expected reward and fault progression cost. One of the main contributions of the paper is the proposed incentive mechanism for power generating units. The proposed incentive mechanism ensures that the power generation will meet the demand while being budget balanced at the convergence point. Furthermore, the proposed incentive mechanism has the property of being individually rational for each of the agents. These two properties make the proposed algorithm particularly attractive for real applications: a) ensuring that the central coordinating system is not imposed with any additional overhead costs; and b) ensuring that it is rational for each of the power generating units to participate in the negotiation process. 

Our simulation results demonstrate that the proposed mechanism results in better performance in comparison to the base-line decisions such as the purely condition-based maintenance or the corrective maintenance.  

It would be also interesting to test the proposed algorithm in larger and more complex networks with more agents and more coordination requirements. 

As future research work we plan to extend the proposed algorithm to electrical markets with high penetration of renewable energy sources, where the amount of power generated by the generators is uncertain. 

\section*{Acknowledgment}
The contribution of Olga Fink was funded by the Swiss National Science Foundation (SNSF) Grant no. PP00P2\_176878. This project is carried out within the frame of the Swiss Centre for Competence in Energy Research on the Future Swiss Electrical Infrastructure (SCCER-FURIES) - Digitalisation programme with the financial support of the Swiss Innovation Agency (Innosuisse-SCCER program).

\bibliographystyle{elsarticle-num}
\bibliography{bib_items}

\end{document}